\documentclass[pdflatex,sn-mathphys]{sn-jnl}

\usepackage{graphicx}%
\usepackage{multirow}%
\usepackage{amsmath,amssymb,amsfonts}%
\usepackage{amsthm}%
\usepackage{mathrsfs}%
\usepackage[title]{appendix}%
\usepackage{xcolor}%
\usepackage{textcomp}%
\usepackage{manyfoot}%
\usepackage{booktabs}%
\usepackage{algorithm}%
\usepackage{algorithmicx}%
\usepackage{algpseudocode}%
\usepackage{listings}
\usepackage{graphicx, mathtools, amsmath, amsthm}
 \usepackage{cite}

\usepackage{lineno}

\newtheorem{corollary}{Corollary}
\newtheorem{lemma}{Lemma}
\newtheorem{defn}{Definition}

\newcommand{\RR}{\mathbb{R}}
\newcommand{\XX}{\mathbb{X}}
\newcommand{\DD}{\mathbb{D}}
\newcommand{\VV}{\mathbb{V}}

\renewcommand{\P}{\mathcal{P}}
\newcommand{\A}{\mathcal{A}}
\newcommand{\C}{\mathcal{C}}
\newcommand{\K}{\mathcal{K}}
\theoremstyle{thmstyleone}%
\theoremstyle{thmstyletwo}%
\newtheorem{example}{Example}%

\theoremstyle{thmstylethree}%

\begin{document}

\title[Article Title]{Flow primitives and infinitesimal generators of Perron-Frobenius and Koopman operators}


\author*[1]{\fnm{Phanindra} \sur{Tallapragada}}\email{ptallap@clemson.edu}

\affil*[1]{\orgdiv{Department of Mechanical Engineering}, \orgname{Clemson University}, \orgaddress{\city{Clemson},  \state{South Carolina}, \country{USA}}}


\abstract{The Koopman and the Perron-Frobenius operators are increasingly becoming popular in the control of complex nonlinear systems such as in a wide variety of robotics problems and flow control. This is in addition to the wide interest in the application of operator methods for better understanding of fluid flows.  A particular problem of relevance to all such applications is, how does the Koopman or the Perron-Frobenius (PF) operator change if the underlying vector field of the dynamical system undergoes small changes or if two vector fields are added. The current numerical methods rely on significant computations and model or parameter changes to the dynamical system often require all the computations to be repeated. This paper reports on a novel approach to the computation of the approximate PF and Koopman operators in such cases. The approach makes use of the exponentials of the infinitesimal generators of these operators. It is shown that this approximation depends on the Lie bracket of the vector field and the perturbation vector field. Examples are described where the Koopman and PF operators are constructed from operators of primitive flows and for cases where the model parameters undergo perturbations. }


\maketitle

\section{Introduction}
The Koopman and Perron Frobenius operator has become a popular tool in recent years for the analysis and control of complex dynamical systems with applications spanning robotics \cite{bruder_RAL_2021, haggerty_science_2023, shi_2024},  aerospace and power grid \cite{korda_ifac_2018} to name a few . This is in addition to the already popular applications in fluid mechanics ranging from microscale flows to geophysical fluid dynamics \cite{froyland_prl_2007, froyland2010transport, rowley_mezic_KMD_2009, mezic_afm_2013}. A significant reason for this popularity is that the dynamics of a nonlinear system can be represented by the linear infinite dimensional Koopman operator that propagates observable functions. Finite dimensional linear approximations of the operator can be constructed and used to develop linear (or bilinear) control systems. This framework can be useful even when the underlying mathematical model of the dynamical system is not fully known but measurements (of observables) is possible, allowing data driven construction of the Koopman operator and control of the system, \cite{korda2018linear,vaidya_cdc_2020, brunton_2021}. The Perron-Frobenius operator, which is the adjoint of the Koopman operator, propagates densities (of states) of a dynamical system and has been used extensively to understand transport in dynamical systems spanning ocean and atmospheric flows to microfluidics and more recently has been used to construct linear high dimensional control systems for manipulating collections of particles in microscale flows.

The current methods like DMD, EDMD and its variations, box counting for calculating the approximate Koopman or PF operators require significant computations with a large number of basis functions and initial conditions of the dynamical system \cite{rowley_mezic_KMD_2009, korda_mezic_2017, brunton22_modernkoopman}. When the underlying vector field of the dynamical system changes, weather due to changes in model parameters or changes to the model itself, these computations have to be performed again. A particular change to the vector field is the addition of another vector field. Specifically suppose $\dot{x} = g(x)$ and $\dot{x}=h(x)$ with $x \in \mathbb{R}^n$ are two dynamical systems and we consider the combined vector field $\dot{x}= g(x) + h(x)$. The PF operator $\P^t$ and the Koopman operator $\K^t$ for time $t$ for this combined system can be approximately computed based on the respective operators associated with the individual vector fields $g(x)$ and $h(x)$.  This approximation relies on the representation of these operators by the exponential of the infinitesimal generators. This approach is similar in spirit to the  Baker-Campbell-Hausdorff formula, \cite{kirillov_2017}. It is shown that this approximation is accurate upto $\mathcal{O}(t^3)$ and that the leading error term is determined by the nested Lie bracket $[(g-h), [g,h]]$.  

Several applications will benefit from this approach of calculating the PF and Koopman operators. One example is in determining and controlling the distribution of ensembles of particles in a planar microscale flow through spatially fixed rotors that can be placed at different locations in the fluid or switched on or off (the so called blinking rotors) that have been investigated in \cite{bt_cdc_2023, bt_pof_2024, bt_physicad_2025}. Suppose the velocity field generated by a single rotor located at $c_1 \in \mathbb{R}^2$ is $g(x)$ and the velocity field generated by a rotor at $c_2\in \mathbb{R}^2$ is $h(x)$ and the combined vector field due to both the rotors is $f=g+h$. This is a specific instance of a combination of flow primitives and the calculation of the transport of ensembles due to a combination of such rotors will be simplified by the proposed approach without requiring a recalculation of the operators. 

Another significance of the results in this paper relates to vector fields dependent on parameters, $\dot{x} = f(x;p)$. The question of the variation Koopman or Perron-Frobenius operators with respect to changes in parameters is important in applications such as robotics. Research in the modeling and control of robot dynamics is seeing increasing application of the Koopman operator framework \cite{bruder_RAL_2021, haggerty_science_2023, shi_2024}. If the parameters of a robot, such as its mass matrix, were to change, the Koopman operator will have to be calculated again. In control co-design the computational cost of this can be prohibitive. In this paper, it is first shown that linear scaling of a vector field leads to a time scaling of the Koopman and PF operators. Next it is shown that small changes in parameters of a system can be approximated as additive changes to a vector field and the Koopman and PF operators can be approximated by the exponentials of a sum of generators of this approximate vector field.

In several past papers the infinitesimal generator approach was used to calculate the Perron Frobenius operator or the Koopman operator. For example \cite{Froyland2024} used the Strang splitting method to efficiently calculate the Perron Frobenius operator and \cite{Froyland2020} considered the evolution of densities by a vector field perturbed by noise (a Wiener process). The observation that control affine systems have control affine Koopman generators have been made in several papers, see for example \cite{goswami_paley_2017, rowley_siam_2020, rowley_siam_2024}. In \cite{rowley_siam_2020, rowley_siam_2024} the Koopman infinitesimal generator approach was used to formulate a bilinear control in the lifted space of observables. The present paper builds on the idea of using infinitesimal generators to efficiently calculate the PF and Koopman operators for sums of vector fields using a modification of Strang splitting. A key new result is the proof that for a sum of vector fields, the operators commute of the vector fields commute. A series expansion for the operator computation is derived and numerical simulations are used to demonstrate these upto order $\mathcal{O}(t^3)$. More importantly the paper makes the connection of the approach of using infinitesimal generators to efficiently calculate the PF and Koopman operators in the case where a dynamical system has parameter variations.

\section{PF operator and infinitesimal generator}\label{sec:1}
Consider a dynamical system 
\begin{equation}\label{eq:vf}
    \dot{x} = f(x)
\end{equation}
with $x\in \mathcal{M} \subset \mathbb{R}^n$. The flow map generated by the vector field is given by the map $\Phi_t : \mathcal{M} \mapsto \mathcal{M}$ 
\begin{equation}\label{eq:flow}
    \Phi_t: x(0) \mapsto x(t).
\end{equation}
Let $(\mathcal{M}, \mathcal{B}, \mu)$ be the standard measure space on $\mathcal{M}$ with Lebesgue measure $\mu$. The Perron-Frobenius (PF) operator $\mathcal{P} : L^1(\mathcal{M},\mathcal{B}, \mu) \mapsto L^1(\mathcal{M},\mathcal{B}, \mu)$ is defined as 
\begin{equation}\label{eq:PF1}
   \int_B  \rho(t,x) d\mu =\int_B \mathcal{P}^t \rho(0,x_0) d\mu = \int_{\Phi^{-1}_t(B) } \rho(0,x_0) d\mu
\end{equation}
where $x_0 = \Phi^{-1}_t(x(t))$, $\rho \in L^1(\mathcal{M})$ and $B \in \mathcal{B}$; see for example \cite{lasota_1994, cvitanovic_chaosbook}. The Perron Frobenius operator propagates densities in time according to the flow of the dynamical system. It is therefore convenient to represent $\rho(x,t) \coloneqq \mathcal{P}^t \rho(x,0)$. Using the change of variables formula
\begin{equation}\label{eq:jacobian}
\int_{\Phi^{-1}_t(B) } \rho(0,x_0) d\mu = \int_B \rho(0,\Phi^{-1}_t(x)) |det(J^{-1}_t(x_0))| d\mu
\end{equation}
where $J_t(x_0) = \frac{d\Phi_t}{dx}$ represents the Jacobian of the flow map evaluated at $x_0 = x(0)$. Combining \eqref{eq:PF1} and \eqref{eq:jacobian} and putting $x_0 = \phi^{-1}_t(x(t))$
\begin{equation}\label{eq:density_evolve}
    \rho(t,x) = \frac{\rho (0,x_0)}{|det(J_t(x_0))|}.
\end{equation}
The right hand side of the \eqref{eq:density_evolve} evokes the formula for the integral of a delta distribution. Suppose $\delta(\cdot) \in L^1(\mathcal{M})$ denotes the Dirac-distribution such that $\int_B \delta(y)d\mu = 1$ if $x \in B$ and $0$ otherwise. Then right hand side of  \eqref{eq:density_evolve} can be rewritten as
\begin{equation}\label{eq:delta}
    \frac{\rho (0,x_0)}{|det(J_t(x_0))|}=  \int_B \delta(x-\Phi_t(x_0)) \rho(0,x_0) d\mu.
\end{equation}
The PF operator can now be rewritten using \eqref{eq:PF1} and \eqref{eq:delta} as 
 \begin{equation}\label{eq:PF1new} 
      \rho(t,x)=   \int_B \delta(x-\Phi_t(x_0)) \rho(0,x_0) d\mu \coloneqq \mathcal{P}^t \rho(0,x_0),
\end{equation}
i.e.
 \begin{equation}\label{eq:PF2}
     \mathcal{P}^t =   \delta(x-\Phi_t(x_0))
\end{equation}
with the propagation of a density represented compactly as $\rho(t) = \mathcal{P}^t \rho(0)$ by dropping the dependence of $\rho$ on $x$. The PF operator satisfies the semigroup property since $\mathcal{P}^{t+s} = \mathcal{P}^t\circ \mathcal{P}^s$ and satisfies the group property if the flow $\phi_t$ is invertible. This motivates a definition for the infinitesimal generator for the group, $\mathcal{A}$ defined as 
\begin{equation} \label{eq:gen}
    \mathcal{A} \rho = \lim_{t \to 0}\frac{(\mathcal{P}^t - \mathcal{I})\rho}{t}
\end{equation}
where $\mathcal{I}$ is the identity operator $\mathcal{I}\circ \rho = \rho$ for all $\rho\in L^1(\mathcal{M})$.
A formal definition for the action of the infinitesimal generator for the PF operator on $\rho \in L^1(\mathcal{M})$ can be given as \cite{lasota_1994}
\begin{equation} \label{eq:gen1}
    \mathcal{A} \rho = -\sum_{i=1}^n \frac{\partial (\rho f_i)}{\partial x_i}.
\end{equation}
The evolution of a density $f$ satisfies the continuity equation
\begin{equation}\label{eq:continuity}
    \left(\frac{\partial }{\partial t} - \mathcal{A} \right) \rho = 0
\end{equation}
where 
\begin{equation}
    \mathcal{A}  = -\frac{\partial f_i}{\partial x_i} -\sum_{i=1}^n f_i\frac{\partial}{\partial x_i} = -\nabla \cdot f - (f) \cdot \nabla.
\end{equation}
The solution to the continuity equation \eqref{eq:continuity} is  
\begin{equation}\label{eq:PF_exp_def}
 \rho(t) =    \mathcal{P}^t \rho(0) = e^{t\mathcal{A}} \rho(0)
\end{equation}

The $L^p$ adjoint to the PF operator is the Koopman operator   $\mathcal{K}_t : L^{\infty}(\mathcal{M}) \mapsto L^{\infty}(\mathcal{M})$ is defined as 
\begin{equation} \label{eq:koopman}
    \mathcal{K}_t \theta(x_0) = \theta(\Phi_t(x_0))
\end{equation}
for $\psi \in L^{\infty}(\mathcal{M})$ is called an observable function.
The infinitesimal generator for the Koopman operator is defined in a manner similar to that of $\mathcal{A}$,
\begin{equation}\label{eq:k_gen}
    \mathcal{C} \theta = \lim_{t \to 0}\frac{\mathcal{K}^t\theta - \theta}{t} = \lim_{t \to 0}\frac{\theta(\Phi_t)-\theta}{t} = \nabla \theta \cdot f.
\end{equation} 
Equation \eqref{eq:k_gen} is also the definition of the time derivative of the observable as
\begin{equation}\label{eq:k_pde}
\left(\frac{d}{dt} -\mathcal{C} \right) \theta(x(t)) = 0,
\end{equation}
the solution to which can be expressed as 
\begin{equation}\label{eq:k_exp_def}
    \theta(x(t)) = e^{t\mathcal{C}}\theta(x(0))
\end{equation}
where $\C = (f)\cdot\nabla$ is the infinitesimal generator. The infinitesimal generator of the Koopman operator can also be thought of as the Lie Derivative, $L_f$ of the obervable $\theta$, along the vector field $f$, which expressed in coordinates is $\nabla \theta \cdot f$,
\[
\frac{d \theta}{d t} = L_f \theta = \nabla \theta \cdot f.
\]

\section{Exponentials and Commutativity} \label{sec:exponentials}
A formal expansion of the exponentials in \eqref{eq:PF_exp_def} and \eqref{eq:k_exp_def} in terms of a power series  can be written giving  a power series definition for the PF and the Koopman operators. The PF operator is therefore
\begin{equation}\label{eq:Pf-exp}
\mathcal{P}^t = e^{t\mathcal{A}}= \sum_{k=0}^{\infty} \frac{t^k}{k!}\mathcal{A}^k,
\end{equation}
with terms like $\mathcal{A}^k$ defined recursively as $\mathcal{A}^k \rho = \mathcal{A} \mathcal{A}^{k-1} \rho$. For example
\[
\begin{split}
\mathcal{A}^2 &= (-\nabla \cdot f -\sum_{i=1}^n f_i\frac{\partial}{\partial x_i})\cdot(-\nabla \cdot f -\sum_{j=1}^n f_j\frac{\partial}{\partial x_j})\\
&=(\nabla \cdot f)^2 + 2(\nabla \cdot f)\sum_{i=1}^n f_i\frac{\partial}{\partial x_i}  + \sum_{i=1}^n\sum_{j=1}^n\left(f_if_j \frac{\partial^2}{\partial x_i \partial x_j} +  f_i \frac{\partial f_j}{\partial x_i}\frac{\partial }{\partial x_j}\right).
\end{split}
\]

Analogously for the Koopman operator
\begin{equation}\label{eq:k_exp}
    \mathcal{K}^t = e^{t\mathcal{C}} = \sum_{k=0}^{\infty} \frac{t^k}{k!}\mathcal{C}^k
\end{equation}
with terms like $\mathcal{C}^k$ defined recursively as $\mathcal{C}^k \theta = \mathcal{C} \mathcal{C}^{k-1} \theta$. For example
\[\begin{split}
\mathcal{C}^2 \theta &= \mathcal{C} \mathcal{C} \theta = \mathcal{C} \frac{\partial \theta}{\partial x_i} f_i = \frac{\partial }{\partial x_j} \left(\frac{\partial \theta}{\partial x_i} f_i\right) f_j\\
&= \sum_{i=1}^n\sum_{j=1}^n \left( \frac{\partial^2 \theta}{\partial x_i \partial x_j} f_i f_j + \frac{\partial \theta}{\partial x_i}\frac{\partial f_i}{\partial x_j} f_j \right)
\end{split}
\]

\subsection{Sum of vector fields}
Consider the dynamical systems $ \dot{x} = g(x)$ with $x \in \mathbb{R}^n$
with the flow map $\psi_t$ and  $\dot{x} = h(x)$, with $ x \in \mathbb{R}^n$ with the flow map $\phi_t$. The associated infinitesimal generators $\mathcal{A}_1 = -\nabla \cdot g -\sum_{i=1}^n g_i\frac{\partial}{\partial x_i}$  and $\mathcal{A}_2=-\nabla \cdot h -\sum_{i=1}^n h_i\frac{\partial}{\partial x_i}$ and Koopman infinitesimal generators are $\C_1 = g \cdot \nabla $ and $\C_2 = h \cdot \nabla$.
The vector field $f = g+h$ creates a new dynamical system 
\begin{equation}\label{eq:sum}
\dot{x} = f(x) = g(x) + h(x)
\end{equation}
with the flow map $\Phi_t$. The infinitesimal generator of the PF operator is
\begin{equation}\label{eq:sumA}
     \mathcal{A} = -\nabla \cdot(g+h)-\sum_{i=1}^n (g_i+h_i)\frac{\partial }{\partial x_i} \\
     = \mathcal{A}_1 + \mathcal{A}_2
\end{equation}
which is the sum of the infinitesimal generators due to the vector fields $g$ and $h$. Similarly the infinitesimal generator of the Koopman operator is
\begin{equation}\label{eq:sumC}
\C = f \cdot \nabla = (g+h) \cdot \nabla = \C_1 + \C_2.
\end{equation}

The PF  operator $\P$, and the Koopman operator $\K$  can be obtained by exponentiating the generators
\begin{equation}\label{eq:PFsum}
    \P^t = e^{t\mathcal{A}}= e^{t(\mathcal{A}_1+\mathcal{A}_2)}
\end{equation}
and
\begin{equation}\label{eq:Ksum}
    \K^t = e^{t\mathcal{C}}= e^{t(\mathcal{C}_1+\mathcal{C}_2)}
\end{equation}
In general the exponential $e^{t(\mathcal{A}_1+\mathcal{A}_2)} \neq e^{t\mathcal{A}_1}e^{t\mathcal{A}_2}$. 

\begin{lemma}
The PF operator associated with the dynamical system \eqref{eq:sum} is $\mathcal{P}^t = \mathcal{P}_1^t\mathcal{P}_2^t$  if $\mathcal{A}_1\mathcal{A}_2 = \mathcal{A}_2\mathcal{A}_1$. The Koopman operator associated with the dynamical system is $\mathcal{K}^t = \mathcal{K}_1^t\mathcal{K}_2^t$  if $\mathcal{C}_1\mathcal{C}_2 = \mathcal{C}_2\mathcal{C}_1$.
\end{lemma}
\begin{proof}
\[
\begin{split}
\mathcal{P}^t &= e^{t(\mathcal{A}_1+\mathcal{A}_2)} = \sum_{k=0}^{\infty}\frac{t^k}{k!}(\mathcal{A}_1+\mathcal{A}_2)^k \\
 &=  \sum_{k=0}^{\infty}\sum_{m=0}^{k}\begin{pmatrix}k\\m \end{pmatrix}\frac{\mathcal{A}_1^m\mathcal{A}_2^{k-m}}{k!}t^k \quad \textrm{if} \quad \mathcal{A}_1\mathcal{A}_2 = \mathcal{A}_2\mathcal{A}_1 \\
 &=\sum_{k=0}^{\infty}\sum_{m=0}^{k}\frac{\mathcal{A}_1^m\mathcal{A}_2^{k-m}}{m!(k-m)!}t^k = \sum_{k,m=0}^{\infty}\frac{(\mathcal{A}_1^mt^m)(\mathcal{A}_2^{k}t^k)}{m!k!}\\
 &=e^{t\mathcal{A}_1}e^{t\mathcal{A}_2}\\
 &=\mathcal{P}_1^t\mathcal{P}_2^t.
\end{split}
\]
The proof for the case of the Koopman operator follows along a similar line.
\end{proof}

\begin{defn}
Given two vector fields $g(x) \in \mathbb{R}^n$ and $h(x) \in \mathbb{R}^n$ and associated PF infinitesimal generators $\mathcal{A}_1 = (-\nabla \cdot g - g \cdot \nabla)$ and $\mathcal{A}_2 = (-\nabla \cdot h - h \cdot \nabla)$ respectively, the commutator  of the PF infinitesimal generators is defined as $[\mathcal{A}_1,\mathcal{A}_2] =\mathcal{A}_1\mathcal{A}_2 - \mathcal{A}_2\mathcal{A}_1$. If the associated Koopman infinitesimal generators are $\mathcal{C}_1 =  g \cdot \nabla$ and $\mathcal{C}_2 =  h \cdot \nabla$ respectively, the commutator  of the Koopman infinitesimal generators is defined as $[\mathcal{C}_1,\mathcal{C}_2] =\mathcal{C}_1\mathcal{C}_2 - \mathcal{C}_2\mathcal{C}_1$ .
\end{defn}

\begin{lemma}
The commutators $[\mathcal{A}_1,\mathcal{A}_2] = [g,h] \cdot \nabla$, where $[g,h]$ and $[\mathcal{C}_1,\mathcal{C}_2] = [g,h] \cdot \nabla$ where $[g,h]$ represents the Lie bracket of the vector fields $g$ and $h$.
\end{lemma}
\begin{proof} A direct calculation gives 
\begin{equation}
\begin{split}
[\mathcal{A}_1,\mathcal{A}_2] &=  (\nabla\cdot g)(\nabla \cdot h) + (\nabla \boldsymbol{\cdot} g) h \boldsymbol{\cdot} \nabla + g\boldsymbol{\cdot} \nabla(\nabla \boldsymbol{\cdot} h) + g \boldsymbol{\cdot} \nabla(h\boldsymbol{\cdot}  \nabla) \\
& -\left((\nabla\cdot g)(\nabla \cdot h) + (\nabla \boldsymbol{\cdot} h)g.\nabla + h \boldsymbol{\cdot} \nabla (\nabla \boldsymbol{\cdot} g) + h \boldsymbol{\cdot} \nabla (g \boldsymbol{\cdot} \nabla)\right) \\
& = g \boldsymbol{\cdot} \nabla(h\boldsymbol{\cdot}  \nabla) - h \boldsymbol{\cdot} \nabla (g \boldsymbol{\cdot} \nabla)\\
&=\sum_{i=1}^n\sum_{j=1}^n\left(g_ih_j \frac{\partial^2}{\partial x_i \partial x_j} +  g_i \frac{\partial h_j}{\partial x_i}\frac{\partial }{\partial x_j} - g_jh_i \frac{\partial^2}{\partial x_i \partial x_j} -  h_i \frac{\partial g_j}{\partial x_i}\frac{\partial }{\partial x_j} \right)\\
&=\sum_{i=1}^n\sum_{j=1}^n\left(g_i \frac{\partial h_j}{\partial x_i}\frac{\partial }{\partial x_j} -  h_i \frac{\partial g_j}{\partial x_i}\frac{\partial }{\partial x_j}\right) =[g,h] \boldsymbol{\cdot}   \nabla.
\end{split}
\end{equation}

The calculation for the commutator of the Koopman operators follows similarly 
\begin{equation}
\begin{split}
[\mathcal{C}_1,\mathcal{C}_2]  &=   \sum_{i=1}^n\sum_{j=1}^n \left( \frac{\partial^2 }{\partial x_i \partial x_j} g_i h_j + \frac{\partial }{\partial x_i}\frac{\partial g_i}{\partial x_j} h_j \right) - \sum_{i=1}^n\sum_{j=1}^n \left( \frac{\partial^2 }{\partial x_i \partial x_j} h_i g_j + \frac{\partial }{\partial x_i}\frac{\partial h_i}{\partial x_j} g_j \right) \\
&= \sum_{i=1}^n\sum_{j=1}^n \left( \frac{\partial }{\partial x_i}\frac{\partial g_i}{\partial x_j} h_j - \frac{\partial }{\partial x_i}\frac{\partial g_i}{\partial x_j} g_j   \right)\\
&=  [g,h] \cdot \nabla
\end{split}
\end{equation}
\end{proof}

\begin{corollary}
Two PF operators  $\mathcal{P}_1^t = e^{\mathcal{A}_1t}$ and $\mathcal{P}_2^t = e^{\mathcal{A}_2t}$ associated with the vector fields $g$ and $h$ commutte if and only if $[g,h]=0$.  Two Koopman operators $\K_1^t = e^{\C_1t}$ and $\K_1^t = e^{\C_2t}$ commute if and only if $[g,h] =0$.
\end{corollary}

\begin{example}
    Consider the vector fields
    \[
    g = \begin{pmatrix}
    x_2\\
    -\omega^2 x_1
    \end{pmatrix} \quad \quad \text{and} \quad \quad 
    h= \begin{pmatrix}
    -\frac{1}{\omega^2} x_2\\
    x_1
    \end{pmatrix}.
    \]
    The vector fields are divergence free, and the infinitesimal generators of respective PF operators are $\mathcal{A}_1 = g \boldsymbol{\cdot} \nabla$ and $\mathcal{A}_2 = h \boldsymbol{\cdot} \nabla$. The Lie bracket $[g,h]=0$; therefore $\mathcal{A}_1\mathcal{A}_2 = \mathcal{A}_2\mathcal{A}_1$ and the PF operator for the dynamical system $\dot{x} = g(x)+h(x)$ is $\mathcal{P}^t = e^{t\mathcal{A}_1}e^{t\mathcal{A}_2}$.
\end{example}

\begin{example}
    Consider the dynamical system $\dot{x} = g(x)+h(x) \in \mathbb{R}^2$ where the vector fields,
    \[
    g = \begin{pmatrix}
    -x_1-x_1x_2-x_1^3\\
    -x_2-x_2^2+x_1^2+x_1^4
    \end{pmatrix} \quad \quad \text{and} \quad \quad 
    h= \begin{pmatrix}
    -x_2\\
    x_1+3x_1x_2+x_1^3
    \end{pmatrix}.
    \]
   The vector fields are not divergence free; $\nabla \boldsymbol{\cdot} g \neq 0$ and $\nabla \boldsymbol{\cdot} h \neq 0$, but  a formal calculation shows $[g,h]=0$. Therefore the infinitesimal generators of both the PF and Koopman operators commute,
    \[ [\mathcal{A}_1, \mathcal{A}_2] = [g,h] \cdot \nabla = 0.
    \]
    and
    \[ [\C_1, \C_2] = [g,h] \cdot \nabla = 0.
    \]
    
    Hence the PF and Koopman operators of the dynamical system $\dot{x} = g(x) + h(x)$ are $\mathcal{P}^t = \P_1^t\P_2^t$ and $\K^t= \K_1^t\K_2^t$.
\end{example}

\subsection{Approximation of the PF and Koopman operators for the sum of vector field} \label{sec:approx}
Now suppose $\mathcal{A}_1\mathcal{A}_2 \neq \mathcal{A}_2\mathcal{A}_1$, then the PF operator $\mathcal{P}^t$ for the combined flow can be approximated in terms of the primitive PF operators, $\mathcal{P}_1^t$ and $\mathcal{P}_2^t$ using the exponential of the generator \eqref{eq:Pf-exp} in a manner similar to the well known Baker-Campbell-Hausdorff series and the Lie-Trotter product formula, \cite{trotter_1959, kirillov_2017}.
\[
\mathcal{P}^t = e^{t(\mathcal{A}_1+\mathcal{A}_2)} = 1 + t(\mathcal{A}_1+\mathcal{A}_2) + \frac{t^2}{2}\left( \mathcal{A}_1+\mathcal{A}_2\right)^2 + \frac{t^3}{6}\left(\mathcal{A}_1+\mathcal{A}_2\right)^3 + ...
\]
Next a direct calculation shows
\[
\begin{split}
\mathcal{P}^t_1\mathcal{P}^t_2  = e^{t\mathcal{A}_1}e^{t\mathcal{A}_2} &= 1 + t(\mathcal{A}_1+\mathcal{A}_2) + \frac{t^2}{2}\left( [\mathcal{A}_1, \mathcal{A}_2] + \left(\mathcal{A}_1+\mathcal{A}_2\right)^2 \right)   \\
&+\frac{t^3}{6}\left(\mathcal{A}_1+\mathcal{A}_2\right)^3 + \frac{t^3}{4}\left((\mathcal{A}_1+\mathcal{A}_2) [\mathcal{A}_1, \mathcal{A}_2] + [\mathcal{A}_1, \mathcal{A}_2](\mathcal{A}_1+\mathcal{A}_2)\right) \\
&+\frac{t^3}{48}\left([\mathcal{A}_1, [\mathcal{A}_1,\mathcal{A}_2]] + [\mathcal{A}_2, [\mathcal{A}_2,\mathcal{A}_1]]\right) + \mathcal{O}(t^4).
\end{split}
\]
Nested bracket of the generators are defined recursively; for example \[
[\mathcal{A}_1, [\mathcal{A}_1,\mathcal{A}_2]] = [(\nabla \cdot g + g \cdot \nabla),[(\nabla \cdot g + g \cdot \nabla), (\nabla \cdot h + h \cdot \nabla)]]
.\]
These nested brackets are once again related to nested Lie brackets of the vector fields, $g$ and $h$, for example
\begin{equation}\label{eq:nested}
[\mathcal{A}_1, [\mathcal{A}_1,\mathcal{A}_2]] = [ g \cdot \nabla,[g,h].\nabla] = [g,[g,h]]\cdot \nabla.
\end{equation}

The operators $\mathcal{P}^t_1\mathcal{P}^t_2$ and $\mathcal{P}^t_2\mathcal{P}^t_1$ can be combined as
\begin{equation}
\begin{split}
\frac{1}{2}\left(\mathcal{P}^t_1\mathcal{P}^t_2 + \mathcal{P}^t_2\mathcal{P}^t_1 \right) &= \frac{1}{2}\left(e^{t\mathcal{A}_1}e^{t\mathcal{A}_2} + e^{t\mathcal{A}_2}e^{t\mathcal{A}_1}\right) \\
&=1 + t(\mathcal{A}_1+\mathcal{A}_2) + \frac{t^2}{2}\left( \mathcal{A}_1+\mathcal{A}_2\right)^2 + \frac{t^3}{6}\left(\mathcal{A}_1+\mathcal{A}_2\right)^3+ \\
&+\frac{t^3}{48}\left([\mathcal{A}_1, [\mathcal{A}_1,\mathcal{A}_2]] + [\mathcal{A}_2, [\mathcal{A}_2,\mathcal{A}_1]]\right) + \mathcal{O}(t^4),
\end{split}
\end{equation}
yielding
\begin{equation} \label{eq:P_cubic}
   \mathcal{P}^t =\frac{1}{2}\left(\mathcal{P}^t_1\mathcal{P}^t_2 + \mathcal{P}^t_2\mathcal{P}^t_1 \right)  + \mathcal{O}(t^3).
\end{equation}

The Koopman operator can be approximated similarly as
\begin{equation}\label{eq:K_cubic}
    \mathcal{C}^t = \frac{1}{2}\left(\mathcal{C}^t_1\mathcal{C}^t_2 + \mathcal{C}^t_2\mathcal{C}^t_1 \right)  + \mathcal{O}(t^3).
\end{equation}

\subsection{Perturbation of parameters}\label{sec:params}

First we will consider the vector field
\begin{equation}\label{eq:scale}
\dot{x} = \epsilon f(x)
\end{equation}
where $\epsilon \in \mathbb{R}$. This is a linear scaling of the vector field $f(x)$. Suppose  $\C_1 = (\nabla)\cdot f$ and $\P_1 = -\nabla \cdot f - (f) \cdot \nabla$ which are the generators of the Koopman and PF operators for the vector field $\dot{x} = f(x)$. The Koopman and PF generators for the scaled vector field \eqref{eq:scale} are $\C_2 = \epsilon\C_1$ and $\A_2 = \epsilon \P_1$ respectively.
The PF operator associated with \eqref{eq:scale} is
\begin{equation}
    \P_2^t = e^{t\A_2} = e^{\epsilon t\A_1} = e^{\tau \A_1} = \P_1^\tau
\end{equation}
where $\tau = \epsilon t$ is a scaled time. Similarly the Koopman operator associated with \eqref{eq:scale} is $\C_2^t = \C_1^\tau$.

Next suppose a vector field $f(x;p) \in \mathbb{R}^n$ depends on the parameters $p \in \mathbb{R}^m$, 
\begin{equation}
    \dot{x} = f(x;p)
\end{equation}
with  nominal values of the parameters being $p_0 \in \mathbb{R}^m$ and the nominal vector field being $\dot{x} = f(x;p_0)$. We consider a change to the parameters $p = p_0 + \epsilon q$ where $\epsilon \in \mathbb{R}$ and $q \in \mathbb{R}^m$. Further assuming that the vector field varies smoothly with respect to the parameters, 
\[
f(x;p) = f(x;p_0) + \epsilon q_j\frac{\partial f_i(x;p_0)}{\partial p_j} + \mathcal{O}(\epsilon^2).
\]
Putting $h(x) = q_j\frac{\partial f_i(x;p_0)}{\partial p_j}$ and $g(x) = f(x;p_0)$, one retrieves the model of \eqref{eq:sum}. If the vector field $f$ is divergence free, then so is the vector field $g(x)= f(x;p_0)$. The vector field $h(x) = q_j\frac{\partial f_i(x;p_0)}{\partial p_j}$ is also divergence free, since $\nabla \cdot h = \frac{\partial h_i}{\partial x_i} = \frac{\partial}{\partial x_i}\left(q_j\frac{\partial f_i(x;p_0)}{\partial p_j} \right) = q_j \frac{\partial}{\partial p_j}\left( \frac{\partial f_i(x;p_0)}{\partial x_i}\right) = 0$. 

Suppose the PF and Koopman operators associated with the vector field $g(x;p_0)$ are $\mathcal{P}_1^t$ and $\mathcal{C}_1^t$ and the operators associated with the vector field $h(x;q,p_0)$ are $\mathcal{P}_2^t$ and $\mathcal{C}_2^t$ then the PF operator associated with the vector field $f(x;p)$ is
\begin{equation}\label{eq:PF_param}
    \mathcal{P}^t \approx \frac{1}{2}\left(\mathcal{P}_1^t\mathcal{P}_2^{\epsilon t} + \mathcal{P}_2^{\epsilon t} \mathcal{P}_1^t\right).
\end{equation}
The Koopman operator associated with the vector $f(x;p_0)$ is 
\begin{equation}\label{eq:Koopman_param}
    \mathcal{C}^t \approx \frac{1}{2}\left(\mathcal{C}_1^t\mathcal{C}_2^{\epsilon t} + \mathcal{C}_2^{\epsilon t} \mathcal{C}_1^t\right).
\end{equation}

\section{Examples  }\label{sec:examples}
The results in section \ref{sec:exponentials} are agnostic to the specific numerical technique to compute the PF and Koopman operators. In this section, examples are presented where the operators are calculated using the the EDMD technique, which is briefly reviewed here.

The Perron-Frobenius operator is calculated by first calculating the Koopman operator using the extended dynamic mode decomposition with a dictionary of basis functions and making use of the adjoint relation between the two operators, as in \cite{klus_jcd_2016}. This method is briefly reviewed here. Suppose  $\DD$ is a dictionary of  $k$ basis functions 
\[
\DD = \{\psi_1, \psi_2, \dots, \psi_k\}
\]
and define $\VV$ as the linear space spanned by $\DD$.  Each $\psi_i$ is a scalar-valued function acting on the state space, i.e. $\psi_i:\mathcal{M} \mapsto\RR$ for $i = 1, \dots, k$, where $\XX$ is the state space. An approximation of the observable $\theta$ is by a linear combination of these basis functions 
\begin{equation}
g(x) \approx \sum_{i = 1}^k \psi_i(x)c_i = \Psi^T(x) c = c^T\Psi(x)
\end{equation}
where $\Psi:\XX\mapsto\RR^k$ is a column-vector valued function where the elements are given by $[\Psi(x)]_i = \psi_i(x)$, and $c\in\RR^k$ is a column vector of coefficients. If the flow map of the dynamical system is $\Phi_t:\mathcal{M} \mapsto \mathcal{M} $, the observable is propagated by the Koopman operator $\mathcal{K}$ and approximated by the projected action of a finite dimensional operator $K$ on the space $\VV $.
\begin{align}
\theta(\Phi_t(x)) &= \K^t \theta(x) \approx [\K^t(c^T\Psi)](x) \approx c^TK_t^T\Psi(x)
\label{eq:Kapproxdef}
\end{align}
To construct the matrix $K$, a set of initial state values $x_i$ distributed uniformly in the domain are chosen and these are evolved forward by time $t$ , $y_i = \Phi_t(x_i)$ for $i = 1, \dots, m$.  We define (snapshot) matrices 
\begin{equation*}
X = 
\begin{bmatrix} 
x_1 & ~\cdots~ & x_m 
\end{bmatrix}
\qquad 
\text{and} 
\qquad 
Y = 
\begin{bmatrix} 
y_1 & ~\cdots~ & y_m 
\end{bmatrix}
\end{equation*}

Using Eq. \eqref{eq:Kapproxdef} along with the relationship $y_i = \Phi(x_i)$, for $i = 1, \dots, m$, $\theta(y_i) \approx \Psi^T(y_i) c = \Psi^T(x_i) K c$. Usually $i>k$, i.e. we choose more measurements than basis functions and so the problem is overdetermined and we compute $K$ as the solution to the least squares problem 
\begin{equation}
    K = \arg \min_K \|\Psi_Y - K^T\Psi_X\|^2_F
\end{equation}
where $\|\cdot\|_F$ is the Frobenius norm.  The analytical solution of this problem is 
\begin{equation}
    K^T = \Psi_Y\Psi_X^\dagger
\end{equation}
where $\Psi_X^\dagger$ is the pseudoinverse of $\Psi_X$.

The Perron-Frobenius operator is adjoint to the Koopman operator, 
\begin{equation}
    \int_{\XX} [\K \theta](x) \rho(x) dx = \int_{\XX} \theta(x) [\P \rho](x)dx.
    \label{eq:adjoint}
\end{equation}

If the density function $\rho$ is also expressed as a linear combination of the same basis functions, $rho(x) \approx \sum_i \psi_i(x) b_i = \Psi^T(x) b = b^T\Psi(x)$, the the action of $\mathcal{P}$ on $\rho$ is approximated by the action of the finite dimensional PF operator $P_t$ on the basis functions;   
\[
\P[b^T\Psi](x) \approx b^TP^T\Psi(x) = \Psi^T(x)Pb
\]
Substituting these approximations into Eq. \eqref{eq:adjoint}, 
\begin{equation}
    \int_{\mathcal{M}} c^TK^T\Psi(x) \Psi^T(x)b \, dx 
    = 
    \int_{\mathcal{M}} c^T\Psi(x) \Psi^T(x)Pb\,dx.
\end{equation}
Since the domain of integration is the same, the integrands must be equal, and so, for each of the data points, $x_i$ in the snapshot matrices, we should have
\[
c^TK^T\Psi(x_i) \Psi^T(x_i)b  
    =  
c^T\Psi(x_i) \Psi^T(x_i)Pb.
\]

So, the matrix approximation $P$ of the Perron-Frobenius operator $\P$ is given by the solution of the least squares problem 
\begin{equation}
    P = \arg \min_{P} \left\| \Psi_Y\Psi_X^T - \Psi_X\Psi_X^TP\right\|^2_F
\end{equation}
which has the analytical solution 
\begin{equation}
    P = (\Psi_X\Psi_X^T)^\dagger\Psi_Y\Psi_X^T
\end{equation}
where $\|\cdot\|_F$ is the Frobenius norm and $(\cdot)^\dagger$ is the pseudoinverse.

\subsection{Flow primitives and PF operator} \label{sec:rotlet}
Consider the following vector field composed of two vector fields $g(x_1,x_2) + h(x_1,x_2)$
\begin{equation}\label{eq:rotlet}
\begin{pmatrix}
    \dot{x}_1 \\
    \dot{x}_2
\end{pmatrix} 
= \underbrace{\begin{pmatrix}
- \frac{\Gamma_1 (y-y_1)(2r_1^2 + \epsilon^2)}{(r_1^2 + \epsilon^2)^{5/2}}\\
  \frac{\Gamma_1 (x-x_1)(2r_1^2 + \epsilon^2)}{(r_1^2 + \epsilon^2)^{5/2}}
\end{pmatrix}}_{g(x_1,x_2)} + 
\underbrace{\begin{pmatrix}
- \frac{\Gamma_2 (y-y_2)(2r_2^2 + \epsilon^2)}{(r_2^2 + \epsilon^2)^{5/2}}\\
  \frac{\Gamma_2 (x-x_2)(2r_2^2 + \epsilon^2)}{(r_2^2 + \epsilon^2)^{5/2}}
\end{pmatrix}}_{h(x_1,x_2)}
\end{equation}
where $r_1^2 = (x-x_1)^2+(y-y_1)^2$ and $r_2^2 = (x-x_2)^2+(y-y_2)^2$.
The vector field $g$ is generated by a rotlet located at $(x_1,y_1)$ with strength $\Gamma_1$ and $h$ is generated by a rotlet located at $(x_2,y_2)$ with strength $\Gamma_2$. The rotlet fields are regularized by the small parameter $\epsilon$ to avoid singularities in the velocity field \cite{ainley_jcp_2008, cortez_jcp_2015}. Transport of densities in a planar fluid flow at low Reynolds number is an important problem that has been addressed using differential dynamic programming in \cite{bt_cdc_2023, bt_pof_2024, bt_physicad_2025}. This is an example of a flow composed of flow primitives where the evolution of an initial density $\rho_0$ (of tracer particles) is of interest. This evolution of the density is calculated using the action of the PF operator on $\rho_0$. Using the EDMD technique a dictionary consisting of  radial basis functions is created with their means evenly distributed in a grid spanning $-4<x<4$ and $-4<y<4$. These basis functions are evaluated on a grid of $161\times161$ points evenly spaced within $-4\leq x \leq 4$ and $-4\leq y \leq 4$. The vector field \eqref{eq:rotlet} is integrated for these initial conditions for a time period $T=1$. The semi group of PF and Koopman are calculated at intervals of $t=0.1$. These numerics are performed for the combined vector field $f$ in \eqref{eq:rotlet} and separately for the vector fields $g$ and $h$ giving the operators $P^{nt}$ associated with the vector field $f$ and $P_1^{nt}$, $P_2^{nt}$ associated with the vector fields $g$ and $h$ respectively, for $n = 0, 1, ..., 10$.

\begin{figure}
\begin{center}
    \includegraphics[width=1\linewidth]{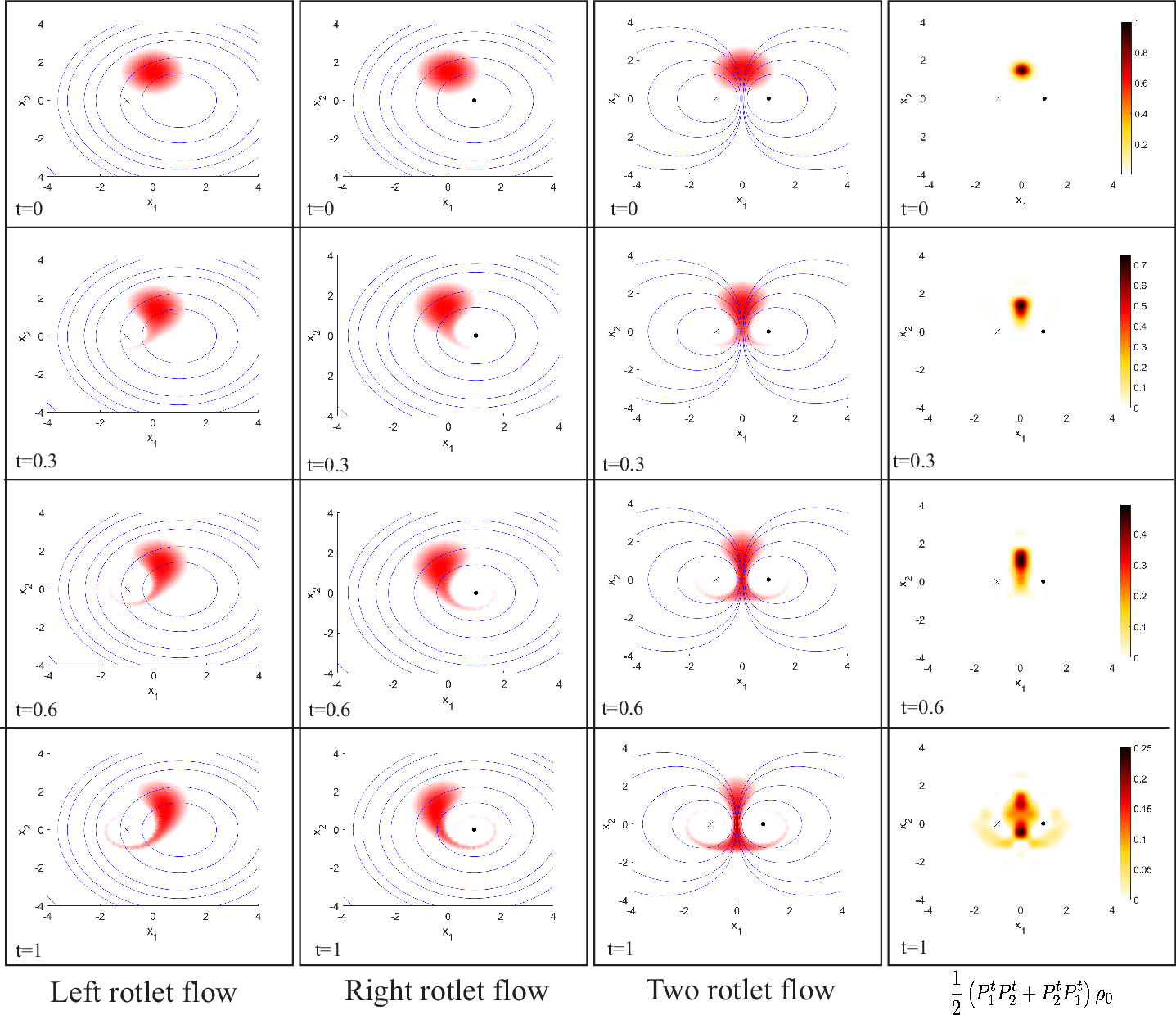}    
    \caption{The first three columns of figures show the evolution of a blob of tracers. (a) The first column shows the evolution of a  blob of tracers  due to the vector field $g$ produced by a single rotlet at $(1,0)$, (b) the second column shows the evolution of a  blob of tracers  due to a vector field $h$ produced by a single rotlet at $(-1,0)$, (c)  the third column shows the evolution of a  blob of tracers  due to a vector field $g+h$ produced by two rotlets at $(-1,0)$ and $(1,0)$ respectively. (d) The fourth column shows the evolution of a (gaussian) density function with its center initially located at the initial center of the blob. This density is propagated by $\frac{1}{2}(P_1^tP_2^t + P_2^tP_1^t)$}
    \label{fig:rotlet}
    \end{center}
\end{figure}

The first column in fig. \ref{fig:rotlet} shows the propagation of a sample density by the flow of $\dot{x} = g$ due to a rotlet at $(1,0)$ with strength $\Gamma_1=1$ and the second column in fig. \ref{fig:rotlet} shows the propagation of the same initial density by  the flow of $\dot{x} = h$ due to a rotlet at $(-1,0)$ with strength $\Gamma_2=1$. The third column in fig. \ref{fig:rotlet} shows the propagation of the same initial density by the flow of $\dot{x} = g+h$ due to two rotlets at $(1,0)$ and $(-1,0)$ with strength $\Gamma_1=1$ and $\Gamma_2=1$. The evolutions of these initial densities is obtained by a direct numerical integration of $1600$ initial conditions sampled by gaussian function. These respresent the evolution of a blob of tracer particles by the three flows respectively. The last column shows the propagation of the same density by the approximate operator $\frac{1}{2}(P_1^tP_2^t + P_2^tP_1^t)$. The evolution of the density by combining the individual PF operators using \eqref{eq:P_cubic} yields accurate results, comparable to the direct integration of a blob of tracers, for integration times as long as $t=1$.


\subsection{Examples - Perturbation of parameters} \label{sec:examples2}
Consider the vector field $g(x_1,x_2;p) \in \mathbb{R}^2$ 
\begin{equation}\label{eq:pendulum}
\begin{pmatrix}
    \dot{x}_1 \\
    \dot{x}_2
\end{pmatrix} 
= {\begin{pmatrix}
 x_2\\
  -p_0^2\sin{x_1}
\end{pmatrix}}
\end{equation}
governing the motion of a simple pendulum, with the parameter $p=p_0$. If the parameter  is perturbed to $p = p_0+\epsilon$, then the vector field can be approximated as
\begin{equation}\label{eq:pen_approx}
\begin{pmatrix}
    \dot{x}_1 \\
    \dot{x}_2
\end{pmatrix} =
\underbrace{\begin{pmatrix}
 x_2\\
  -(p_0+\epsilon)^2\sin{x_1}
\end{pmatrix}}_{f(x_1,x_2; p)}
 \approx
\underbrace{\begin{pmatrix}
 x_2\\
  -p_0^2\sin{x_1}
\end{pmatrix}}_{g(x_1,x_2; p_0)} + \underbrace{\begin{pmatrix}
 0\\
  -2\epsilon p_0\sin{x_1}
\end{pmatrix}}_{h(x_1,x_2)}.
\end{equation}
Around a nominal value of the parameter $p_0$, the vector field is $g(x_1,x_2,p_0)$. If the parameter $p=p_0+\epsilon$, then $h_i = \epsilon \frac{\partial f_i(x;p_0)}{\partial p}$. The propagation of a density or an observable by this dynamical system with the perturbed parameters can be calculated using \eqref{eq:PF_param}. 

Figure  \ref{fig:pendulum} shows the propagation of an initial density $\rho(x_1,x_2)$ (a gaussian) due to the flow generated by the vector field $f(x;p)$ \eqref{eq:pen_approx}. The first column shows the density at $t=0$ and its evolution at $t=0.5$ and $t=1$ obtained by a direct integration of $1450$ initial conditions sampled using $\rho(x_1,x_2)$. The second column shows the propagation of the density $\rho$ using the PF operator associated with $f(x;p)$. The third column shows the propagation of the density using the PF operator $\frac{1}{2}(P_1^tP_2^t + P_2^tP_1^t)$ where $P_2^t$ is the PF operator associated with $\epsilon h(x;p_0)$.
The PF operators were calculated using the EDMD method described in \ref{sec:rotlet}, with 225 radial basis functions with their means evenly distributed in a grid spanning $-4<x<4$ and $-4<y<4$.

\begin{figure}[!h]
\begin{center}
    \includegraphics[width=0.95\linewidth]{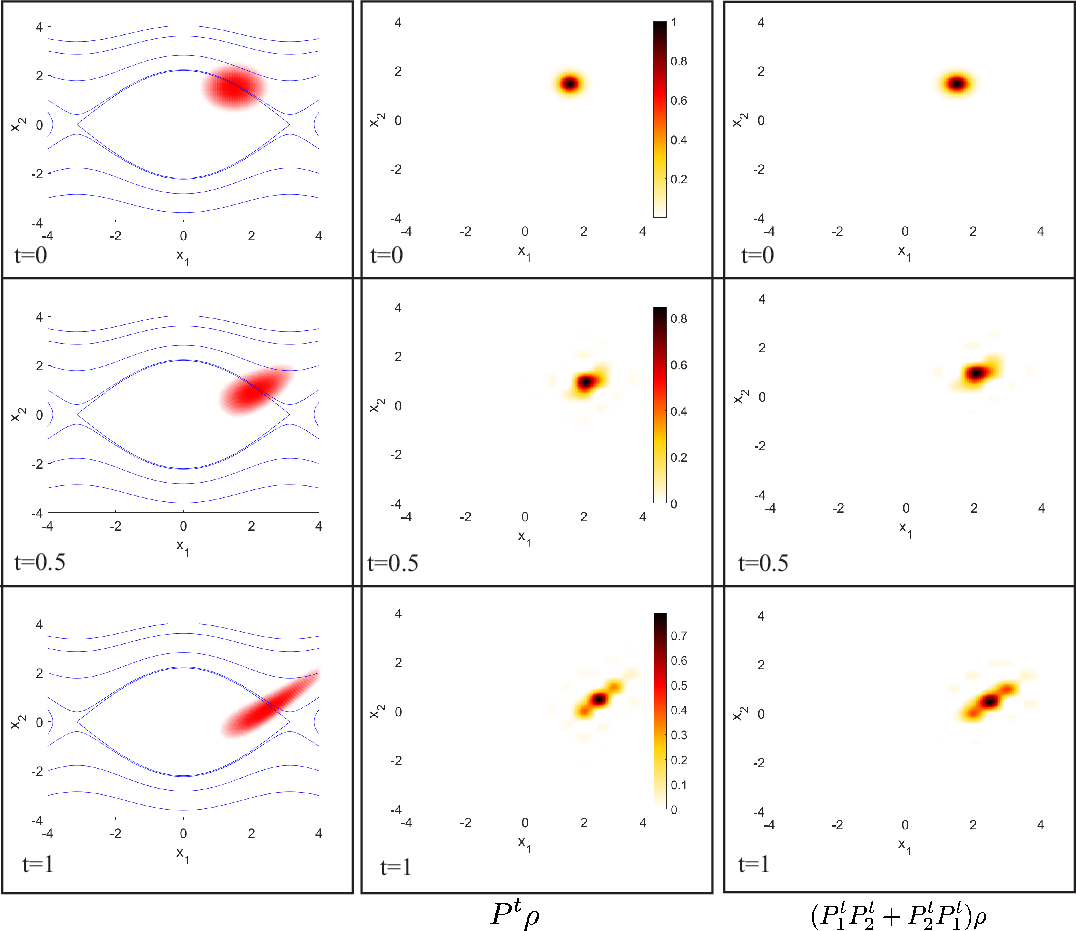}    
    \caption{(a) The first column shows the evolution of a  blob of tracers  due to the exact the vector field $f(x;p)$ \eqref{eq:pen_approx}. (b) The second column shows the propagation of a  gaussian density function by the PF operator $P^t$ due to the vector field $g+h$. (c)  The third column shows the propagation of a  gaussian density function by the PF operator $\frac{1}{2}(P_1^tP_2^t + P_2^tP_1^t)$ where the operator $P_1^t$ is due to vector field $g$ and the operator $P_2^t$ due to the vector field $h$. The parameter values are $p_0=1$ and $\epsilon = 0.1$.}
    \label{fig:pendulum}
    \end{center}
\end{figure}

A second example further demonstrates the computation of the PF operator for a dynamical system with parameters that are perturbed. Consider the vector field of the Duffing oscillator with parameter $p = p_0 + \epsilon$,
\begin{equation}\label{eq:duffing}
\begin{pmatrix}
    \dot{x}_1 \\
    \dot{x}_2
\end{pmatrix} 
= \underbrace{\begin{pmatrix}
 x_2\\
 x_1 + (p_0 + \epsilon) x_1^3
\end{pmatrix}}_{f(x_1,x_2;p)} =
\underbrace{\begin{pmatrix}
 x_2\\
   x_1 + p_0 x_1^3
\end{pmatrix}}_{g(x_1,x_2;p_0)}
+
\underbrace{\begin{pmatrix}
 0\\
  \epsilon x_1^3
\end{pmatrix}}_{h(x_1,x_2; \epsilon )}.
\end{equation}
Suppose $p_0=0$ and $\epsilon <0$; then the vector field $g$ has a saddle fixed point at $(0,0)$ around which almost all trajectories diverge. The vector field $h$ has  degenerate set of fixed points $x_1=0$. The  PF operator $\mathcal{P}^t$ associated with the vector field $f(x;p)$ can be approximated by $\frac{1}{2}\left(P_1^tP_2^t+P_2^t P_1^t \right)$ where $\mathcal{P}_1^t$ is associated with the vector field $g(x;p_0)$ and $\mathcal{P}_2^t$ is associated with $h(x;\epsilon )$. Figure fig. \ref{fig:duffing} shows the propagation of an initial density $\rho(x_1,x_2)$ (a gaussian) due to the flow generated by the vector field $f(x;p)$ \eqref{eq:pen_approx}. The first column shows the density at $t=0$ and its evolution at $t=0.5$ and $t=1$ obtained by a direct integration of $1450$ initial conditions sampled using $\rho(x_1,x_2)$. The second column shows the propagation of the density $\rho$ by $\mathcal{P}^t$. The third column shows the propagation of the density using the PF operator $\frac{1}{2}(P_1^tP_2^t + P_2^tP_1^t)$. The PF operators were calculated using the EDMD method described in \ref{sec:rotlet}, with 225 radial basis functions with their means evenly distributed in a grid spanning $-4<x<4$ and $-4<y<4$.

\begin{figure}[!h]
\begin{center}
    \includegraphics[width=0.9\linewidth]{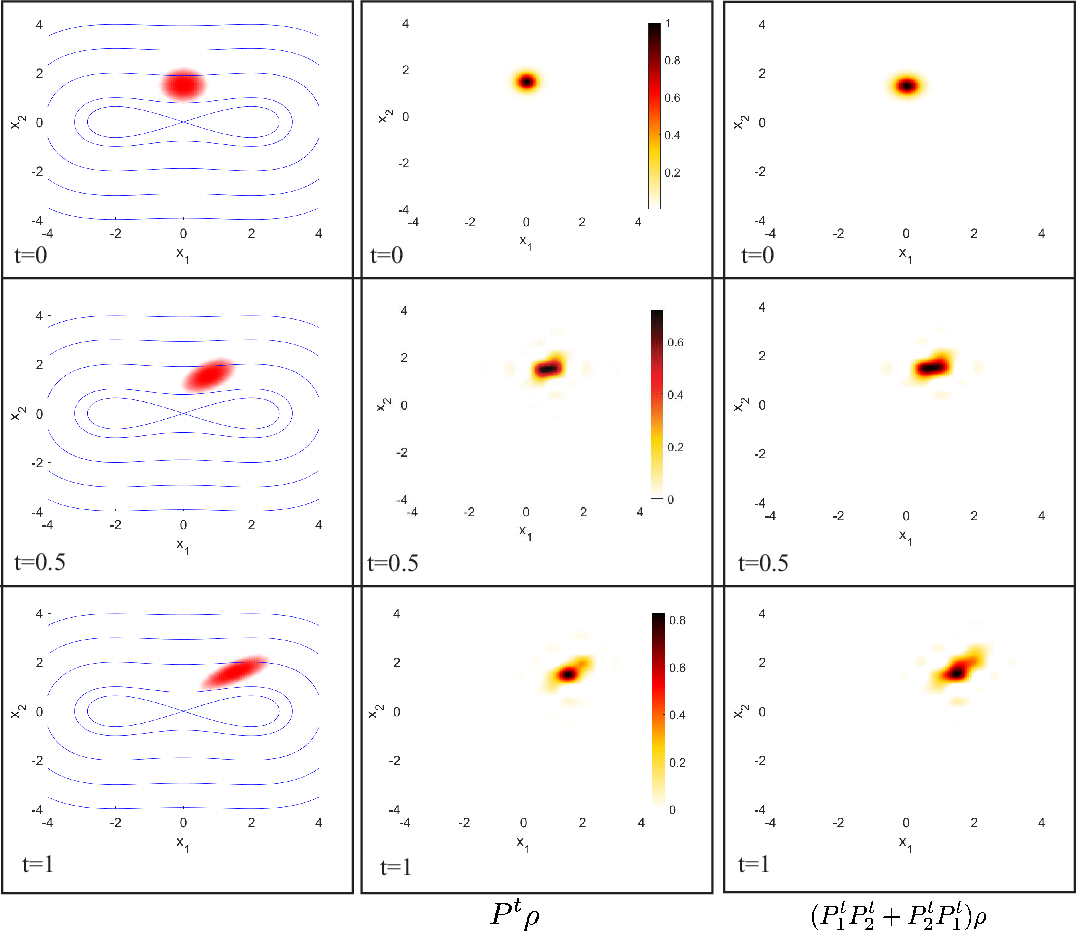}    
    \caption{(a) The first column shows the evolution of a  blob of tracers  due to the vector field $g+h$ \eqref{eq:duffing}. (b) The second column shows the propagation of a  gaussian density function by the PF operator $P^t$ due to the vector field $g+h$. (c)  The third column shows the propagation of a  gaussian density function by the PF operator $\frac{1}{2}(P_1^tP_2^t + P_2^tP_1^t)$ where the operator $P_1^t$ is due to vector field $g$ and the operator $P_2^t$ due to the vector field $h$.}
    \label{fig:duffing}
    \end{center}
\end{figure}

In theory it is possible to compute the PF and Koopman operators using the exponential of the generators, for instance $P^t = e^{(A_1+A_2)t}$, where the generators $A_1$ and $A_2$ are finite dimensional approximations of $\mathcal{A}_1$ and $\mathcal{A}_2$ associated with the vector field $g$ and $h$. The computation of $e^{(A_1+A_2)t}$ using scaling and squaring method; however the errors for this calculation can be unknown for general matrices $A_1$ and $A_2$ \cite{moler_siam2003, higham_siam2005} that are associated with the vector fields $g$ and $h$. Figure \ref{fig:compare_exp} shows the propagation of the density functions shown in fig. \ref{fig:rotlet} and fig. \ref{fig:pendulum} by $P^t = e^{(A_1+A_2)t}$. The density propagated through the computation of the matrix exponential shows significant errors compared to the density propagation by $\frac{1}{2}\left(P_1^tP_2^t + P_2^tP_1^t \right)$ or compared to the direct numerical simulations of a blob of initial conditions in the third column of fig. \ref{fig:rotlet} or the first column of fig. \ref{fig:pendulum}. The matrix exponential was calculated using MATLAB's matrix exponential function which uses the scaling and squaring method.

\begin{figure}[!h]
\begin{center}
    \includegraphics[width=0.8\linewidth]{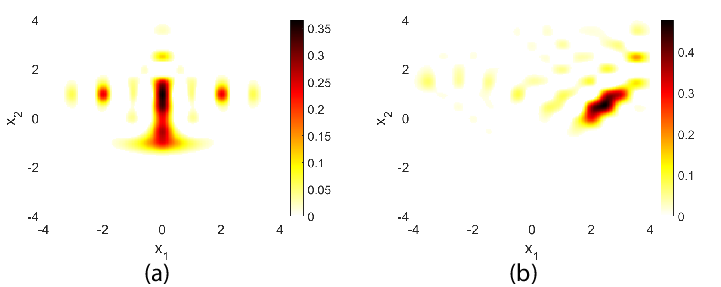}    
    \caption{The propagation of the density functions shown in fig. \ref{fig:rotlet} and fig. \ref{fig:pendulum} at $t=1$ by $P^t = e^{(A_1+A_2)t}$. (a) Density propagated by the generators due to the two rotlets vector fields, $g$ and $h$ in \eqref{eq:rotlet}, (b)  Density propagated by the generators due to the vector fields $g$ and $h$ in \eqref{eq:pen_approx}.}
    \label{fig:compare_exp}
    \end{center}
\end{figure}
\section{Conclusion}
The current methods for the calculation of the Koopman and Perron Frobenius operators for a dynamical (or control system) require significant computations with a large number of initial conditions. When the dynamical system undergoes even small changes, due to a change of model parameters, or due to changes in the model itself, these computations have to be repeated. The results in this paper provide a framework to approximate the PF and Koopman operators from a library of operators associated with a class of vector fields. This framework is agnostic to the specific technique of calculating the Koopman or PF operators; the EDMD method is used in this paper to demonstrate key results, but any other method to calculate the operators could be used. The error in the approximation increases with time according to \eqref{eq:P_cubic}. Future work can address the improvements to the accuracy for longer time horizons. The analysis and control of complex dynamical systems using the Koopman and Perron Frobenius operators are becoming increasingly popular especially in robotics, especially where the dynamics of a robot are complex or the interaction with the environment is complex. In such applications, while a short horizon prediction of an observable or density is sufficient, model parameter changes or changes to environmental inputs such as hydrodynamic forces on a swimming robot or ground contact force due changes in terrain on a legged or wheeled robot impose a significant cost of recomputing the operators. It is in such applications that the results of this paper can be expected to have the most significant impact.




\bibliographystyle{unsrt}

\end{document}